\documentclass[11pt,a4paper]{article}
\usepackage[margin=1in]{geometry}
\usepackage[T1]{fontenc}
\usepackage{lmodern}
\usepackage{amsmath,amssymb,amsthm,mathtools}
\usepackage{booktabs,array}
\usepackage{microtype}
\usepackage{xcolor}
\usepackage{tikz}
\usepackage[numbers,sort&compress]{natbib}
\usepackage[colorlinks=true,linkcolor=blue!45!black,
 citecolor=blue!45!black,urlcolor=blue!45!black]{hyperref}
\newtheorem{theorem}{Theorem}[section]
\newtheorem{lemma}[theorem]{Lemma}
\newtheorem{corollary}[theorem]{Corollary}
\theoremstyle{definition}

\newcommand{\ind}[1]{\mathbf{1}[#1]}
\newcommand{\dep}{\operatorname{dep}}
\newcommand{\opt}{\operatorname{opt}}
\newcommand{\calI}{\mathcal{I}}

\newcommand{\calQ}{\mathcal{Q}}
\newcommand{\EQ}{\mathsf{EQ}}
\newcommand{\NEQ}{\mathsf{NEQ}}
\newcommand{\Z}{\mathbb{Z}}
\title{Partial Interval Multicover Is NP-Complete\\Even for Unit Intervals}
\author{Mingchao Zhou\textsuperscript{1},~Yingli Ran\textsuperscript{1},~
Lei Zhao\textsuperscript{1},~Xiangzhi Tu\textsuperscript{1},~Zhao Zhang\textsuperscript{1}\thanks{Corresponding author. Email: \href{mailto:hxhzz@sina.com}{\texttt{hxhzz@sina.com}}}.\\
\normalsize\textsuperscript{1}Zhejiang Normal University}
\date{}

\begin{document}
\maketitle

\begin{abstract}
For a  interval family $\calI$, a  point set $X\subseteq\mathbb R$, demands $d:X\to\Z_{\ge0}$, and $p\in\Z_{\ge0}$, partial interval multicover asks for a subfamily $S\subseteq\calI$ minimizing $|S|$ subject to
\[
 |\{x\in X:\dep_S(x)\ge d(x)\}|\ge p,
 \qquad \dep_S(x)=|\{I\in S:x\in I\}|.
\]
We prove that the unweighted decision problem is NP-complete even for unit intervals, answering the open question in~\cite{Li2026}.
\end{abstract}

\section{Introduction}
\label{sec:introduction}

An instance of unweighted partial interval multicover~\cite{RanJinZhang2025,Li2026} is a tuple $(\calI,X,d,p)$, where $\calI=\{I_1,\ldots,I_N\}$ is a family of $N$  intervals, $X\subseteq\mathbb R$ is a finite set of distinct points, $d:X\to\Z_{\ge0}$ gives the number of selected intervals needed to fully cover each point, and $p\in\Z_{\ge0}$ is the minimum number of points that must be fully covered. For a selected subfamily $S\subseteq\calI$, define
\[
\dep_S(x)=|\{I\in S:x\in I\}|.
\]
We say that $S$ \emph{fully covers} $x$ if $\dep_S(x)\ge d(x)$. The objective is to minimize $|S|$ while fully covering at least $p$ points of $X$.

The decision version additionally takes a budget $B\in\Z_{\ge0}$ and asks whether some subfamily $S\subseteq\calI$ satisfies
\begin{equation}
 |S|\le B,
 \qquad |\{x\in X:\dep_S(x)\ge d(x)\}|\ge p.
 \label{eq:decision}
\end{equation}
Each interval may be selected at most once. All intervals have unit cost, and each fully covered point is counted once.

\subsection{Related work.}
Slav{\'i}k~\cite{Slavik1997} analyzed the approximation ratio of the greedy algorithm for partial cover. In set multicover, each element has an integer demand and must belong to at least that many selected sets. Rajagopalan and Vazirani~\cite{RajagopalanVazirani1998} gave randomized parallel approximation algorithms for this problem. Chekuri et al.~\cite{Chekuri2012} studied geometric set multicover, where the elements are points and the covering sets are geometric regions.

In partial multicover, an element is counted only when its demand is met. Shi et al.~\cite{Shi2019} gave an approximation algorithm that fully covers at least $(1-\varepsilon)p$ elements when $p$ are required, where $\varepsilon\in(0,1)$. Ran et al.~\cite{Ran2021} gave an algorithm that fully covers at least $p$ elements, with an approximation ratio independent of the maximum demand.

When $d(x)=1$ for every $x\in X$, this is partial interval cover, which is polynomial-time solvable by the dynamic program of Golab et al.~\cite{Golab2009}. With arbitrary demands, the known approximation bounds include a $2$-approximation by Ran et al.~\cite{RanJinZhang2025} and a polynomial-time approximation scheme by Li et al.~\cite{Li2026}. The budgeted maximization version, called maximum interval multicover, instead seeks a subfamily $S$ with $|S|\le B$ that fully covers as many points as possible; Xu et al.~\cite{XuRanZhang2025} gave a $1/2$-approximation for that problem.

\subsection{Our result.}
\begin{theorem}
	\label{thm:main}
	Unweighted partial interval multicover is NP-complete even when all intervals have unit length.
\end{theorem}

\paragraph{Reduction outline.}
We reduce from independent set on cubic planar graphs. In Subsection~\ref{sec:grid}, we represent vertex choices by bits and replace graph edges by grid paths. For any fixed choice of the original bits, choosing the new bits optimally adds the same constant to the original reward. Subsection~\ref{sec:windows} represents each grid bit by selected positions in a sequence. A block of consecutive positions earns a reward when it contains a specified number of selected positions. Large rewards ensure that every optimal selection of $B$ positions represents the grid bits correctly; this part of the reward is the same for every choice of the grid bits. Subsection~\ref{sec:intervals} represents positions by equal-length intervals and blocks by points, so that the selected count in each block equals the coverage depth at its point. Subsection~\ref{sec:demands} turns the exact-count rewards into coverage demands by adding a linear term. Equal interval lengths ensure that this term contributes the same total whenever exactly $B$ intervals are selected, so it does not change which selection is best. Finally, Subsection~\ref{sec:hardness} proves that the two decision instances have the same answer and that the construction has polynomial size.

\section{Proof of NP-completeness}
\label{sec:reduction}

\subsection{Binary rewards on a grid}
\label{sec:grid}

Our starting problem is independent set on simple 2-connected cubic planar graphs, which is NP-hard by Mohar~\cite{Mohar2001}. Here cubic means that every vertex has degree three. Let $(G,t)$ be an instance, with $G=(V,E)$, $n=|V|$, $m=|E|$, and $1\le t\le n$. Let $\alpha(G)$ be the size of a maximum independent set. Write $\ind{P}=1$ when the statement $P$ is true and $\ind{P}=0$ otherwise. We must decide whether $\alpha(G)\ge t$.

For bits $x_v\in\{0,1\}$, $v\in V$, consider the objective
\begin{equation}
 J(x)=\sum_{uv\in E}\ind{x_u\ne x_v}
       +\sum_{v\in V}\ind{x_v=0}.
 \label{eq:binary-reward}
\end{equation}
Let $S=\{v:x_v=1\}$ and $e(S)=|E(G[S])|$. Since $G$ is cubic, exactly $3|S|-2e(S)$ edges have exactly one endpoint in $S$. Thus
\begin{equation}
 J(x)=3|S|-2e(S)+n-|S|
     =n+2\bigl(|S|-e(S)\bigr).
 \label{eq:binary-identity}
\end{equation}

\begin{lemma}
\label{lem:independence}
For every  simple graph $G$,
\[
 \max_{S\subseteq V(G)}\bigl(|S|-e(S)\bigr)=\alpha(G).
\]
Consequently, the maximum of~\eqref{eq:binary-reward} on a cubic graph is $n+2\alpha(G)$.
\end{lemma}
\begin{proof}
	If $S$ is a maximum independent set, then $e(S)=0$ and
	$|S|-e(S)=\alpha(G)$. Thus the maximum is at least $\alpha(G)$.
	For any $S\subseteq V(G)$, we show that $|S|-e(S)\le\alpha(G)$.
	Whenever $e(S)>0$, choose a vertex $v\in S$ incident with an
	edge of $G[S]$ and set $S'=S\setminus\{v\}$. Then
	$$
	|S'|-e(S')
	=|S|-e(S)+\deg_{G[S]}(v)-1
	\ge |S|-e(S).
	$$
	Repeating this step leaves an independent set $S^*$ with
	$$
	|S|-e(S)\le |S^*|\le\alpha(G).
	$$
	
	As $x$ ranges over all binary assignments, $S=\{v:x_v=1\}$
	ranges over all subsets of $V(G)$. For a cubic graph,
	\eqref{eq:binary-identity} gives
	$$
	\max_x J(x)
	=n+2\max_{S\subseteq V(G)}\bigl(|S|-e(S)\bigr)
	=n+2\alpha(G).
	$$
	The Lemma is proved.
\end{proof}

\paragraph{Replacing edges by grid paths.}

Draw $G$ with vertices at integer grid points and edges formed by horizontal and vertical segments, so that edge interiors are disjoint. Such a drawing with $O(n)$ bends can be found in polynomial time~\cite{TamassiaTollis1989}. Number the columns containing vertices or bends consecutively from left to right, and do the same for the rows from bottom to top. Since there are $O(n)$ vertices and bends, the drawing now lies in an $O(n)\times O(n)$ grid. Add a vertex at every intermediate integer grid point along each path, dividing it into edges of length one. There are $O(n)$ segments, each of length $O(n)$, so the total number of unit edges is $O(n^2)$. A related grid construction for independent set on cubic planar graphs appears in~\cite{Wocjan2003}.

Keep the bits at the original vertices and assign a bit $x_u\in\{0,1\}$ to each added vertex. For each grid edge $ab$, define $\EQ(a,b)=\ind{x_a=x_b}$ and $\NEQ(a,b)=\ind{x_a\ne x_b}$. If the path replacing an original edge $uv$ has $\ell$ edges, put an $\NEQ$ reward on the first edge starting at $u$ and an $\EQ$ reward on each of the remaining $\ell-1$ edges. All rewards have weight one.

In the example shown in Figure~\ref{fig:grid-path}, the original edge $uv$ is replaced by a path of four unit edges. The black vertices $u$ and $v$ are the original endpoints, and the three white vertices are added along the path. The edge from $u$ to the first added vertex receives an $\NEQ$ reward, and the other three edges receive $\EQ$ rewards. If $x_u=0$, $x_v=1$, and all three added vertices have bit $1$, all four rewards are obtained, giving a total reward of $4$.

\begin{figure}[htbp]
\centering
\begin{tikzpicture}[x=1.3cm,y=1.3cm,line width=0.8pt,font=\small]
  \useasboundingbox (-0.6,-1) rectangle (8.2,2.2);

  \draw (0,0) -- (1.5,0);
  \fill (0,0) circle (2.3pt);
  \fill (1.5,0) circle (2.3pt);
  \node at (0,-0.5) {$u$};
  \node at (1.5,-0.5) {$v$};
  \node at (0.75,0.5) {$\NEQ$};

  \draw[->,>=stealth] (2.1,0) -- (3.1,0);

  \draw[gray!25,very thin,step=1] (4,0) grid (7,1);
  \draw (4,0) -- (5,0) -- (6,0) -- (6,1) -- (7,1);
  \fill (4,0) circle (2.3pt);
  \fill (7,1) circle (2.3pt);
  \foreach \p in {(5,0),(6,0),(6,1)}
    \draw[fill=white] \p circle (2.3pt);
  \node at (3.5,0.5) {$u$};
  \node at (7.6,1) {$v$};
  \node at (4.5,-0.5) {$\NEQ$};
  \node at (5.5,-0.5) {$\EQ$};
  \node at (6.5,0.5) {$\EQ$};
  \node at (6.5,1.5) {$\EQ$};
\end{tikzpicture}
\caption{Edge replacement.}
\label{fig:grid-path}
\end{figure}
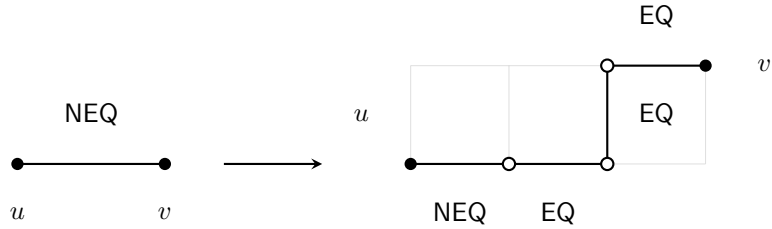

\begin{lemma}
\label{lem:path}
For fixed endpoint bits $x_u,x_v$, the maximum total reward on the path replacing $uv$, over all choices of its internal bits, is $ \ell-1+\ind{x_u\ne x_v}.$
\end{lemma}
\begin{proof}
Set every internal bit to $x_v$. Each of the $\ell-1$ $\EQ$ rewards then equals $1$, and the $\NEQ$ reward on the first edge equals $\ind{x_u\ne x_v}$.

If $x_u\ne x_v$, this assignment obtains all $\ell$ rewards, so the maximum is $\ell$.

If $x_u=x_v$, obtaining all $\EQ$ rewards would require every internal bit to equal $x_v=x_u$. The first edge would then contribute $0$. Thus no assignment can obtain all $\ell$ rewards, and the total reward is at most $\ell-1$. The assignment above gives reward $\ell-1$, so the maximum is $\ell-1$.
\end{proof}

For each $e\in E$, let $\ell_e$ be the number of edges in its grid path, and let $Q=\sum_{e\in E}\ell_e$. The term $\ell_e-1$ in Lemma~\ref{lem:path} depends only on the path length. Summing these terms gives
$$
 C=\sum_{e\in E}(\ell_e-1)=Q-m.
$$
Keep the reward $\ind{x_v=0}$ at each original vertex $v$, and assign no such reward to the new vertices.

Fix the bits $x$ at the original vertices. Since the paths have no internal vertices in common, their internal bits can be chosen independently. By Lemma~\ref{lem:path}, the maximum total reward for these fixed original bits is
$$
 \sum_{e=uv\in E}\bigl(\ell_e-1+\ind{x_u\ne x_v}\bigr)
 +\sum_{v\in V}\ind{x_v=0}
 =C+J(x).
$$
Let $\opt_{\mathrm{grid}}$ be the maximum total reward over all grid bits. Maximizing also over the original bits and using Lemma~\ref{lem:independence} gives
\begin{equation}
 \opt_{\mathrm{grid}}
 =C+\max_x J(x)
 =C+n+2\alpha(G).
 \label{eq:grid-value}
\end{equation}

For the four-edge path in Figure~\ref{fig:grid-path}, the maximum reward is $3$ when $x_u=x_v$ and $4$ when $x_u\ne x_v$. The original edge earns $0$ and $1$ in these cases, so replacing it adds $3$ in either case.

We use diagonals as layers so that every grid edge joins two consecutive layers. Shift the drawing so that its leftmost column and lowest row have coordinate zero. Put points with the same $a+b$ in one layer, and number their positions by $a+1$. Thus $(a,b)$ receives the indices
\begin{equation}
 \tau=a+b,\qquad i=a+1.
 \label{eq:grid-indexing}
\end{equation}
A vertical unit edge from $(a,b)$ to $(a,b+1)$ keeps $i$ unchanged and increases $\tau$ by one. A horizontal unit edge from $(a,b)$ to $(a+1,b)$ increases both $\tau$ and $i$ by one. Hence every edge has the form
\begin{equation}
 (\tau,i)(\tau+1,j),\qquad j\in\{i,i+1\}.
 \label{eq:adjacent-layers}
\end{equation}
	Choose $T,K=O(n)$ with $K\ge2$ so that every vertex of the drawing has indices $0\le\tau\le T$ and $1\le i\le K$. For each unused position in this range, add a bit that appears in no reward.

In Figure~\ref{fig:diagonal-numbering}, the point $(1,0)$ has indices $(1,2)$. Its vertical neighbor $(1,1)$ has indices $(2,2)$, and its horizontal neighbor $(2,0)$ has indices $(2,3)$. Both edges join layer $1$ to layer $2$; the position index either stays the same or increases by one.

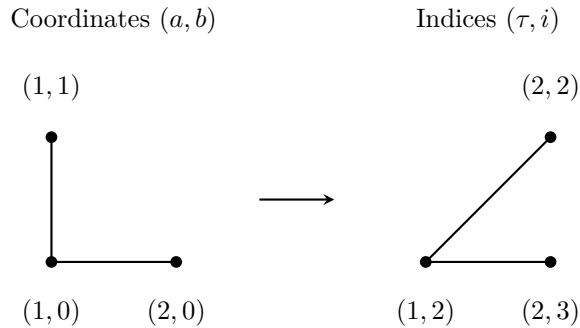
\begin{figure}[htbp]
\centering
\begin{tikzpicture}[x=1.1cm,y=1.1cm,line width=0.8pt,font=\small]
  \useasboundingbox (-1,-1.2) rectangle (7,3.5);
  \draw (0,1.5) -- (0,0) -- (1.5,0);
  \foreach \p in {(0,0),(0,1.5),(1.5,0)}
    \fill \p circle (2.2pt);
  \node at (0,-0.6) {$(1,0)$};
  \node at (0,2.1) {$(1,1)$};
  \node at (1.5,-0.6) {$(2,0)$};
  \node at (0.75,2.9) {Coordinates $(a,b)$};
  \draw[->,>=stealth] (2.5,0.75) -- (3.4,0.75);
  \draw (6,0) -- (4.5,0) -- (6,1.5);
  \foreach \p in {(4.5,0),(6,0),(6,1.5)}
    \fill \p circle (2.2pt);
  \node at (4.5,-0.6) {$(1,2)$};
  \node at (6,-0.6) {$(2,3)$};
  \node at (6,2.1) {$(2,2)$};
  \node at (5.25,2.9) {Indices $(\tau,i)$};
\end{tikzpicture}
\caption{Diagonal numbering.}
\label{fig:diagonal-numbering}
\end{figure}

\subsection{Encoding the grid by consecutive windows}
\label{sec:windows}

We represent each grid bit by selected positions in a sequence. Let $y_h\in\{0,1\}$ indicate whether position $h$ is selected. A window $W=[l,r]$ consists of the consecutive positions $l,\ldots,r$. It receives a reward of weight $w$ when exactly $d$ of its positions are selected:
\begin{equation}
 w\ind{\sum_{h=l}^{r}y_h=d},
 \qquad y_h\in\{0,1\},\quad w,d\in\Z_{>0}.
 \label{eq:exact-reward}
\end{equation}
A window may carry several rewards for different counts.

\paragraph{Selection variables and consistency.}

For each layer $\tau$, introduce two rows of $2K$ binary variables,
$$
\begin{aligned}
 A_\tau&=(a_{\tau,1},a'_{\tau,1},\ldots,a_{\tau,K},a'_{\tau,K}),\\
 B_\tau&=(b_{\tau,1},b'_{\tau,1},\ldots,b_{\tau,K},b'_{\tau,K}),
\end{aligned}
$$
and place them in the order $A_0,B_0,A_1,B_1,\ldots,A_T,B_T$ to obtain $y=(y_1,\ldots,y_N)$. Set
\begin{equation}
 N=4K(T+1),\qquad B=2K(T+1)=N/2.
 \label{eq:selection-size}
\end{equation}
We maximize the sum of the window rewards over all selections satisfying $\sum_{h=1}^N y_h=B$. With $c_\tau=4K\tau$, the variables occupy the positions
$$
\begin{aligned}
 a_{\tau,i}&=y_{c_\tau+2i-1}, & a'_{\tau,i}&=y_{c_\tau+2i},\\
 b_{\tau,i}&=y_{c_\tau+2K+2i-1}, & b'_{\tau,i}&=y_{c_\tau+2K+2i}.
\end{aligned}
$$

We want each pair to contain one selected position, with $a_{\tau,i}$ and $b_{\tau,i}$ taking different values:
\begin{equation}
 a_{\tau,i}+a'_{\tau,i}=b_{\tau,i}+b'_{\tau,i}=1,
 \qquad a_{\tau,i}+b_{\tau,i}=1.
 \label{eq:consistent}
\end{equation}
An assignment satisfying~\eqref{eq:consistent} for every $\tau,i$ is called \emph{consistent}. Each such assignment selects two of the four positions for every grid bit, giving exactly $B$ selections. All assignments of size $B$ are allowed; the rewards below ensure that every optimal assignment is consistent.

We also use $A_\tau$ and $B_\tau$ to denote the full-row windows. For $2\le i\le K$, let $E_{\tau,i}$ start at $a_{\tau,i}$ and end at $b'_{\tau,i-1}$. For $1\le i\le K$, let $F_{\tau,i}$ start at $a'_{\tau,i}$ and end at $b_{\tau,i}$. The rewards on the $E$ windows will make corresponding pairs contain the same number of selected positions. The rewards on the $F$ windows will make their first entries different. Table~\ref{tab:windows} lists their endpoint positions.

\begin{table}[htbp]
\centering
\caption{Consistency windows, with $c_\tau=4K\tau$.}
\label{tab:windows}
\small
\begin{tabular}{@{}lll@{}}
\toprule
Window & Left endpoint & Right endpoint\\
\midrule
Full row $A_\tau$ & $c_\tau+1$ & $c_\tau+2K$\\
$E_{\tau,i}$, $2\le i\le K$ & $c_\tau+2i-1$ & $c_\tau+2K+2i-2$\\
$F_{\tau,i}$, $1\le i\le K$ & $c_\tau+2i$ & $c_\tau+2K+2i-1$\\
Full row $B_\tau$ & $c_\tau+2K+1$ & $c_\tau+4K$\\
\bottomrule
\end{tabular}
\end{table}

On the grid, there are $Q$ rewards for edges and $n$ rewards for original vertices, each at most one. Their sum is at most
\begin{equation}
 U=Q+n.
 \label{eq:reward-bound}
\end{equation}
Set $H=U+1$. Losing a reward of weight $H$ cannot be compensated for by all the grid rewards together. Put a reward of weight $H$ at target count $K$ on each full-row window and each $E_{\tau,i}$. On each $F_{\tau,i}$, put the two rewards
\begin{equation}
 H\ind{z=K-1}+H\ind{z=K+1},
 \label{eq:complement-reward}
\end{equation}
where $z$ is the number of selected positions in that window. Since $z$ cannot equal both $K-1$ and $K+1$, these two rewards together contribute either $0$ or $H$. Count this pair as one group. Each full-row reward and each $E$ reward also forms one group, so there are
\begin{equation}
 M=\bigl(2+(K-1)+K\bigr)(T+1)=(2K+1)(T+1)
 \label{eq:group-count}
\end{equation}
groups. Their total consistency reward is at most $MH$. These weight-$H$ rewards are counted separately from rewards of weight one for the grid.

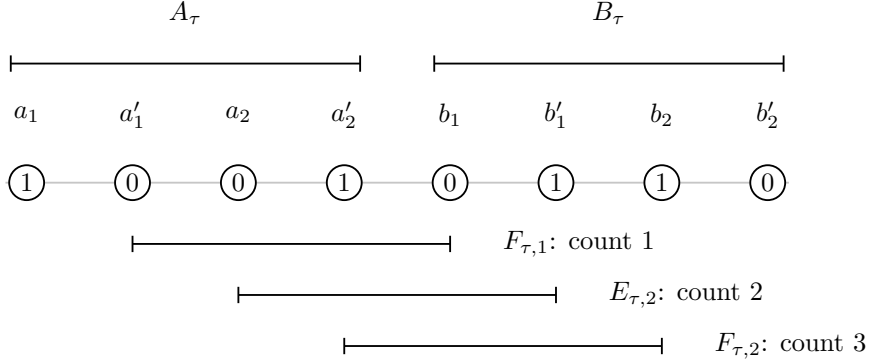
\begin{figure}[htbp]
\centering
\begin{tikzpicture}[x=1.4cm,y=0.9cm,line width=0.7pt,font=\small]
  \useasboundingbox (-0.6,-3.2) rectangle (8.9,3.3);
  \draw[gray!45] (-0.2,0) -- (7.2,0);
  \foreach \x/\name/\value in {
    0/a_1/1,1/{a'_1}/0,2/a_2/0,3/{a'_2}/1,
    4/b_1/0,5/{b'_1}/1,6/b_2/1,7/{b'_2}/0}
  {
    \node at (\x,1.0) {$\name$};
    \node[draw,circle,fill=white,inner sep=1.8pt] at (\x,0) {$\value$};
  }
  \draw (-0.15,1.75) -- (3.15,1.75);
  \draw (-0.15,1.63) -- (-0.15,1.87);
  \draw (3.15,1.63) -- (3.15,1.87);
  \node at (1.5,2.5) {$A_\tau$};
  \draw (3.85,1.75) -- (7.15,1.75);
  \draw (3.85,1.63) -- (3.85,1.87);
  \draw (7.15,1.63) -- (7.15,1.87);
  \node at (5.5,2.5) {$B_\tau$};
  \draw (1,-0.9) -- (4,-0.9);
  \draw (1,-1.02) -- (1,-0.78);
  \draw (4,-1.02) -- (4,-0.78);
  \node[anchor=west] at (4.4,-0.9) {$F_{\tau,1}$: count $1$};
  \draw (2,-1.65) -- (5,-1.65);
  \draw (2,-1.77) -- (2,-1.53);
  \draw (5,-1.77) -- (5,-1.53);
  \node[anchor=west] at (5.4,-1.65) {$E_{\tau,2}$: count $2$};
  \draw (3,-2.4) -- (6,-2.4);
  \draw (3,-2.52) -- (3,-2.28);
  \draw (6,-2.52) -- (6,-2.28);
  \node[anchor=west] at (6.4,-2.4) {$F_{\tau,2}$: count $3$};
\end{tikzpicture}
\caption{Windows for $K=2$.}
\label{fig:consistency-windows}
\end{figure}

In the example in Figure~\ref{fig:consistency-windows}, $K=2$, $A_\tau=(1,0,0,1)$, and $B_\tau=(0,1,1,0)$. The layer index is omitted from the variable labels. Both full rows have count $2$. Window $E_{\tau,2}$ contains the bits $(0,1,0,1)$ and also has count $2$. Windows $F_{\tau,1}$ and $F_{\tau,2}$ contain $(0,0,1,0)$ and $(1,0,1,1)$, respectively, giving counts $1$ and $3$. Each of the five groups in this layer contributes $H$, for a total of $5H$.

\begin{lemma}
\label{lem:consistency}
The consistency reward equals $MH$ if and only if the assignment is consistent.
\end{lemma}
\begin{proof}
Suppose that the consistency reward is $MH$. Every group must then contribute $H$. Fix a layer, omit $\tau$ from its variable names, and write
$$
 P_A(h)=\sum_{i=1}^h(a_i+a'_i),\qquad
 P_B(h)=\sum_{i=1}^h(b_i+b'_i),
 \qquad P_A(0)=P_B(0)=0.
$$
The two full-row rewards give $P_A(K)=P_B(K)=K$. Window $E_i$ contains pairs $i,\ldots,K$ of $A$ followed by pairs $1,\ldots,i-1$ of $B$. Its selected count is therefore
$$
 K-P_A(i-1)+P_B(i-1).
$$
To obtain its reward, this count must equal $K$. Applying this to $E_2,\ldots,E_K$ gives $P_A(h)=P_B(h)$ for $1\le h<K$; the equality also holds at $h=0,K$. Subtracting the equality at $h=i-1$ from the equality at $h=i$ yields
\begin{equation}
 a_i+a'_i=b_i+b'_i\quad(1\le i\le K).
 \label{eq:equal-pairs}
\end{equation}
Window $F_i$ starts at $a'_i$ and ends at $b_i$. Its count is
\begin{equation}
 K-P_A(i-1)-a_i+P_B(i-1)+b_i=K+b_i-a_i.
 \label{eq:F-count}
\end{equation}
For~\eqref{eq:complement-reward} to contribute $H$, this count must be $K-1$ or $K+1$, so $a_i\ne b_i$. If $(a_i,b_i)=(0,1)$, equation~\eqref{eq:equal-pairs} gives $a'_i=1+b'_i$, forcing $(a'_i,b'_i)=(1,0)$. If $(a_i,b_i)=(1,0)$, it gives $1+a'_i=b'_i$, forcing $(a'_i,b'_i)=(0,1)$. Thus~\eqref{eq:consistent} holds for each pair in this layer, and the same argument applies to every layer.

Now suppose that the assignment is consistent. Each pair contains one selected position, so $P_A(h)=P_B(h)=h$ and both full rows have count $K$. Every $E_i$ also has count $K$. Since $b_i-a_i\in\{-1,1\}$, equation~\eqref{eq:F-count} gives count $K-1$ or $K+1$ on every $F_i$. Each of the $M$ groups contributes $H$, giving consistency reward $MH$.
\end{proof}

\paragraph{Recovering the grid rewards.}
\label{subsec:grid-windows}

For a consistent assignment, equation~\eqref{eq:consistent} gives
$$
 (a_{\tau,i},a'_{\tau,i},b_{\tau,i},b'_{\tau,i})
 =(1-b_{\tau,i},b_{\tau,i},b_{\tau,i},1-b_{\tau,i}).
$$
The bit $b_{\tau,i}$ determines all four entries. Define the grid bit represented by these entries as
\begin{equation}
 x_{\tau,i}=b_{\tau,i}\oplus(\tau\bmod2).
 \label{eq:logical-bit}
\end{equation}
Here $\oplus$ is addition modulo two: $x_{\tau,i}=b_{\tau,i}$ in even layers and $x_{\tau,i}=1-b_{\tau,i}$ in odd layers. Equal grid bits in consecutive layers therefore have opposite $b$ values. The two end entries of the window below then sum to one when the grid bits agree, and to zero or two when they differ.

For each grid edge~\eqref{eq:adjacent-layers}, define
\begin{equation}
 W_{\tau,i,j}=[c_\tau+2K+2i,\;c_\tau+4K+2j-1].
 \label{eq:cross-window}
\end{equation}
This window starts at $b'_{\tau,i}$ in row $B_\tau$ and ends at $a_{\tau+1,j}$ in the next row $A_{\tau+1}$. Between these two entries, it contains $K-i$ complete pairs of $B_\tau$ and $j-1$ complete pairs of $A_{\tau+1}$. Each complete pair contributes one selection, so
\begin{align}
 \sum_{h\in W_{\tau,i,j}}y_h
 &= (1-b_{\tau,i})+(K-i)+(j-1)+a_{\tau+1,j}\notag\\
 &= K+j-i+(-1)^\tau(x_{\tau+1,j}-x_{\tau,i}).
 \label{eq:cross-count}
\end{align}
The second equality uses~\eqref{eq:logical-bit}. If $\tau$ is even, then $b_{\tau,i}=x_{\tau,i}$ and $a_{\tau+1,j}=x_{\tau+1,j}$. If $\tau$ is odd, these entries are $1-x_{\tau,i}$ and $1-x_{\tau+1,j}$, respectively.

Put $D=K+j-i\in\{K,K+1\}$. Equation~\eqref{eq:cross-count} gives count $D$ when the endpoint grid bits agree, and count $D-1$ or $D+1$ when they differ. For an $\EQ$ edge, put a reward of weight one at target count $D$. For an $\NEQ$ edge, put a reward of weight one at each of the target counts $D-1,D+1$. A window cannot have both counts, so the two $\NEQ$ rewards together contribute at most one, even for an inconsistent assignment. All target counts are positive because $K\ge2$.

For example, take $K=2$, an even layer $\tau$, and $i=j=1$. The selected count in $W_{\tau,1,1}$ is $2+x_{\tau+1,1}-x_{\tau,1}$:
$$
\begin{array}{c|cccc}
 (x_{\tau,1},x_{\tau+1,1})&(0,0)&(0,1)&(1,0)&(1,1)\\
 \hline
 \displaystyle\sum_{h\in W_{\tau,1,1}}y_h&2&3&1&2
\end{array}
$$
A reward at count $2$ therefore represents an $\EQ$ edge. Rewards at counts $1$ and $3$ represent an $\NEQ$ edge, with exactly one of them obtained when the bits differ.

Each original vertex has the reward $\ind{x_{\tau,i}=0}$. For a consistent assignment, equations~\eqref{eq:F-count} and~\eqref{eq:consistent} give count $K+2b_{\tau,i}-1$ on $F_{\tau,i}$. If $\tau$ is even, $x_{\tau,i}=0$ means $b_{\tau,i}=0$, giving count $K-1$. If $\tau$ is odd, it means $b_{\tau,i}=1$, giving count $K+1$. Add a reward of weight one to this window at count $K-1$ in even layers and at count $K+1$ in odd layers. Only original vertices receive this additional reward. The consistency reward on this window is at most $H$. The additional reward is $0$ or $1$ and is counted among the $n$ original-vertex rewards.

\begin{lemma}
\label{lem:window-value}
Let $\opt_{\mathrm{win}}(B)$ be the maximum total window reward over selections of exactly $B$ positions. Then
\begin{equation}
 \opt_{\mathrm{win}}(B)=MH+C+n+2\alpha(G).
 \label{eq:window-value}
\end{equation}
\end{lemma}
\begin{proof}
For any choice of the grid bits, equation~\eqref{eq:logical-bit} determines $b_{\tau,i}$, and equation~\eqref{eq:consistent} determines the other three entries. The resulting assignment has exactly $B$ selected positions and consistency reward $MH$. By~\eqref{eq:cross-count}, each window for a grid edge gives the same reward as that edge. The added rewards on $F$ windows give the rewards at the original vertices. Its total reward is therefore $MH$ plus the total reward on the grid. Every consistent assignment determines grid bits in this way, so the maximum over consistent assignments is
$$
 MH+\opt_{\mathrm{grid}}=MH+C+n+2\alpha(G)
$$
by~\eqref{eq:grid-value}.

If an assignment is inconsistent, Lemma~\ref{lem:consistency} implies that at least one group contributes $0$. Its consistency reward is at most $(M-1)H$. Each grid edge and each original vertex contributes at most one reward of weight one, so the remaining reward is at most $Q+n=U$. Its total reward is at most
$$
 (M-1)H+U=MH-1.
$$
A consistent assignment of size $B$ exists for every choice of the grid bits and has reward at least $MH$. Hence every optimal assignment with exactly $B$ selections is consistent, and its value is~\eqref{eq:window-value}.
\end{proof}

\subsection{From windows to equal-length intervals}
\label{sec:intervals}

We represent each sequence position $h$ by an interval $I_h$ and each window by a query point. The query point belongs to $I_h$ exactly when its window contains $h$. The windows containing a fixed position must therefore be consecutive in the query order.

\paragraph{Ordering the windows.}

Consider the full-row windows, the $E$ and $F$ windows, and the windows for edges between consecutive layers. A window carrying several rewards is counted only once here. Let $s$ be the number of distinct windows.

\begin{lemma}
\label{lem:window-order}
The windows can be indexed as $W_q=[l_q,r_q]$, $1\le q\le s$, so that both sequences $(l_q)$ and $(r_q)$ are nondecreasing.
\end{lemma}
\begin{proof}
Within layer $\tau$, list the windows in the order
\begin{equation}
 A_\tau,F_{\tau,1},E_{\tau,2},F_{\tau,2},\ldots,
 E_{\tau,K},F_{\tau,K},B_\tau.
 \label{eq:within-layer-order}
\end{equation}
Their left endpoints are $c_\tau+1,\ldots,c_\tau+2K+1$, and their right endpoints are $c_\tau+2K,\ldots,c_\tau+4K$. Both increase along the list.

After $B_\tau$, list the windows~\eqref{eq:cross-window} by increasing $i$, and by increasing $j$ when $i$ is fixed. Since $j\in\{i,i+1\}$, the $j$ values are also nondecreasing. Both endpoint sequences are therefore nondecreasing.

The smallest possible endpoints of these windows are $c_\tau+2K+2$ and $c_\tau+4K+1$, both larger than the corresponding endpoints of $B_\tau$. Their largest possible endpoints are $c_\tau+4K$ and $c_\tau+6K-1$. These are smaller than the endpoints $c_\tau+4K+1$ and $c_\tau+6K$ of $A_{\tau+1}$. Thus the same order works across successive layers. Removing windows for absent grid edges does not disturb either endpoint sequence.
\end{proof}

Fix a position $h$. A window contains $h$ precisely when its right endpoint is at least $h$ and its left endpoint is at most $h$. Define
\begin{equation}
 \alpha_h=\min\{q:r_q\ge h\},\qquad
 \beta_h=\max\{q:l_q\le h\}.
 \label{eq:incidence-bounds}
\end{equation}
Every position belongs to a full-row window, so both indices exist and $\alpha_h\le\beta_h$. Since the right endpoints are nondecreasing, $r_q\ge h$ holds exactly for $q\ge\alpha_h$. Since the left endpoints are nondecreasing, $l_q\le h$ holds exactly for $q\le\beta_h$. Combining the two conditions gives
\begin{equation}
 \{q:h\in W_q\}=\{\alpha_h,\alpha_h+1,\ldots,\beta_h\}.
 \label{eq:consecutive-incidence}
\end{equation}
Thus $\alpha_h$ and $\beta_h$ are the indices of the first and last windows containing $h$. As $h$ increases, the smallest index with $r_q\ge h$ cannot decrease, and neither can the largest index with $l_q\le h$. Hence
$$
 \alpha_h\le\alpha_{h+1},\qquad
 \beta_h\le\beta_{h+1}\qquad(1\le h<N).
$$

Put the query for $W_q$ at $2q$. To make $I_h$ contain exactly the queries from $2\alpha_h$ to $2\beta_h$, choose
\begin{equation}
 I_h=\left[2\alpha_h-1+\frac{h}{N+1},\;
             2\beta_h+\frac{h}{N+1}\right].
 \label{eq:interval-realization}
\end{equation}
Its left endpoint lies strictly between $2\alpha_h-1$ and $2\alpha_h$, and its right endpoint lies strictly between $2\beta_h$ and $2\beta_h+1$. Therefore
\begin{equation}
 2q\in I_h\quad\Longleftrightarrow\quad
 \alpha_h\le q\le\beta_h
 \quad\Longleftrightarrow\quad h\in W_q.
 \label{eq:incidence-preserved}
\end{equation}
For $S(y)=\{I_h:y_h=1\}$, counting the selected intervals at each query gives
\begin{equation}
 \dep_{S(y)}(2q)=\sum_{h\in W_q}y_h,
 \qquad |S(y)|=\sum_{h=1}^Ny_h.
 \label{eq:depth-identity}
\end{equation}
Assign all rewards of $W_q$ to its query at $2q$. A selection of positions and the corresponding selection of intervals have the same total reward and size.

The term $h/(N+1)$ separates endpoints that would otherwise coincide. Because $\alpha_h$ and $\beta_h$ are nondecreasing, both endpoint sequences in~\eqref{eq:interval-realization} are strictly increasing. Every left endpoint lies in an interval $(2u-1,2u)$ and every right endpoint lies in an interval $(2v,2v+1)$, with $u,v\in\Z$. Thus no left endpoint equals a right endpoint. All endpoints are distinct, and no interval contains another.

In the example shown in Figure~\ref{fig:windows-intervals}, take $N=4$ positions and the windows $W_1=[1,2]$, $W_2=[2,3]$, and $W_3=[3,4]$. Formula~\eqref{eq:interval-realization} gives
$$
\begin{aligned}
 I_1&=[1.2,2.2], & I_2&=[1.4,4.4],\\
 I_3&=[3.6,6.6], & I_4&=[5.8,6.8].
\end{aligned}
$$
For position $h=2$, the first and last windows containing it are $W_1$ and $W_2$, so $\alpha_2=1$ and $\beta_2=2$. The interval $I_2=[1.4,4.4]$ contains the queries $2$ and $4$, but not $6$. Selecting positions $2$ and $3$ gives window counts $1,2,1$, and selecting $I_2$ and $I_3$ gives these same depths at $2,4,6$.

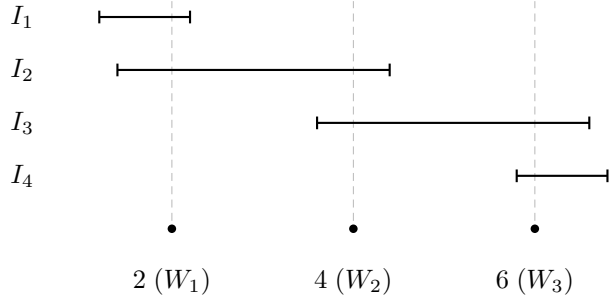
\begin{figure}[htbp]
\centering
\begin{tikzpicture}[x=1.2cm,y=0.7cm,line width=0.8pt,font=\small,
  every node/.style={inner sep=0pt}]
  \useasboundingbox (-0.5,-2.1) rectangle (7.5,5);
  \foreach \q in {2,4,6}
    \draw[gray!55,densely dashed,very thin] (\q,0.15) -- (\q,4.3);
  \foreach \l/\r/\y/\h in {1.2/2.2/4/1,1.4/4.4/3/2,3.6/6.6/2/3,5.8/6.8/1/4}{
    \draw (\l,\y) -- (\r,\y);
    \draw (\l,\y-0.12) -- (\l,\y+0.12);
    \draw (\r,\y-0.12) -- (\r,\y+0.12);
    \node at (0.35,\y) {$I_{\h}$};
  }
  \foreach \q/\k in {2/1,4/2,6/3}{
    \fill (\q,0) circle (1.6pt);
    \node at (\q,-1) {$\q\;(W_{\k})$};
  }
\end{tikzpicture}
\caption{Windows as intervals.}
\label{fig:windows-intervals}
\end{figure}

\paragraph{Integer endpoints and equal lengths.}

The intervals constructed above may have different lengths. Bogart and West~\cite{BogartWest1999} showed that, when no interval contains another, the intervals can be replaced by unit intervals without changing the order of their endpoints. For the next step, we also need integer endpoints and a polynomial bound on the distance between the extreme endpoints.

\begin{lemma}
\label{lem:equal-length}
Suppose $N$ intervals have distinct endpoints, and their left and right endpoints increase in the same order. In polynomial time, they can be replaced by intervals with integer endpoints and common length $L=N+1$, preserving the strict order of all endpoints. The distance between the smallest and largest new endpoints is $O(N^2)$. Query points that are not original endpoints can be moved so that each point is still contained in exactly the same intervals.
\end{lemma}
\begin{proof}
By~\cite{BogartWest1999}, unit intervals $[s_i,s_i+1]$ with the same endpoint order exist. A comparison between endpoints of different intervals requires $s_i+\varepsilon<s_j+\delta$, where $\varepsilon,\delta\in\{0,1\}$ and $i\ne j$. Equivalently, $s_i-s_j<\delta-\varepsilon$, so every required comparison has the form
\begin{equation}
 s_v-s_u<c,\qquad c\in\{-1,0,1\}.
 \label{eq:strict-constraints}
\end{equation}
We can write these $O(N^2)$ inequalities from the original endpoint order. The numbers $s_i$ show that they have a solution; we construct integer coordinates directly from the inequalities.

Introduce integer unknowns $t_1,\ldots,t_N$ and replace each inequality by
\begin{equation}
 t_v-t_u\le Lc-1,\qquad L=N+1.
 \label{eq:integer-constraints}
\end{equation}
Represent this inequality by a directed edge $u\to v$ of weight $Lc-1$. Consider a directed cycle with no repeated vertices and $h\le N$ edges. Summing the corresponding strict inequalities~\eqref{eq:strict-constraints} gives $0<\sum c$. Since this sum is an integer, the total weight of the new cycle is
$$
 L\sum c-h\ge L-N=1.
$$
There is therefore no directed cycle of negative total weight.

Add a new vertex with a zero-weight directed edge to every existing vertex, and let $t_i$ be the shortest distance from the new vertex to $i$. For an edge $u\to v$ of weight $Lc-1$, a shortest path to $u$ followed by this edge gives $t_v\le t_u+(Lc-1)$. Thus the distances satisfy~\eqref{eq:integer-constraints}. They are integers because all edge weights are integers. A shortest path can be chosen with at most $N$ edges, each of absolute weight at most $L+1$, giving
$$
 -N(L+1)\le t_i\le0.
$$

Set $J_i=[t_i,t_i+L]$ and add the same integer to every endpoint to make them nonnegative. Each inequality~\eqref{eq:integer-constraints} preserves the strict order of the corresponding endpoints, with a separation of at least one. Each interval has length $L>0$, so its left endpoint precedes its right endpoint. The distance between the smallest and largest endpoints is $O(N^2)$.

Each query originally lies between two consecutive endpoints or outside all endpoints. Moving it to the corresponding gap in the new endpoint order leaves it in exactly the same intervals. Since the new endpoints are distinct integers, each bounded gap contains a half-integer $h+\tfrac12$, $h\in\Z$, at which we can place the query. Queries outside all endpoints can be kept on the corresponding side. Writing down the inequalities and finding the shortest distances both take polynomial time.
\end{proof}

For the four intervals in Figure~\ref{fig:windows-intervals}, one possible replacement with $L=N+1=5$ is
$$
 J_1=[0,5],\qquad J_2=[3,8],\qquad
 J_3=[7,12],\qquad J_4=[10,15].
$$
These intervals have the same endpoint order as $I_1,I_2,I_3,I_4$. Move the queries $2,4,6$ to $4.5,7.5,11.5$, respectively. The point $4.5$ belongs to $J_1,J_2$, the point $7.5$ to $J_2,J_3$, and the point $11.5$ to $J_3,J_4$, exactly as before. For instance, the original order $1.2<1.4<2.2$ requires $s_1<s_2<s_1+1$. The corresponding integer inequalities are $t_1-t_2\le-1$ and $t_2-t_1\le4$, which are satisfied by $t_1=0,t_2=3$.

Applying Lemma~\ref{lem:equal-length} to the intervals in~\eqref{eq:interval-realization} gives intervals with integer endpoints and a common positive integer length. Each query remains in exactly the same intervals and is not an endpoint. Thus~\eqref{eq:depth-identity} still holds, and the maximum reward with $B$ selections is still~\eqref{eq:window-value}. The equal lengths will let us convert the exact rewards into coverage demands.

\subsection{From exact rewards to coverage demands}
\label{sec:demands}

A reward for depth exactly $d$ is lost when the depth exceeds $d$, whereas a point remains fully covered once its demand is met. We add the same linear term in every integer cell to make the reward functions nondecreasing. We then create demand points so that the number of fully covered points equals the new reward. For selections of exactly $B$ equal-length intervals, the total added term is fixed.

\paragraph{Total coverage depth.}

Consider $N$ intervals with integer endpoints and common length $L\in\Z_{>0}$, a budget $1\le B\le N$, and rewards at query points that are not interval endpoints. Each reward has a positive integer weight and is obtained at a specified positive integer depth. Let $e_{\min}$ and $e_{\max}$ be the smallest and largest endpoints. We call each open interval $(h,h+1)$, $h\in\Z$, a cell. The midpoints $h+\tfrac12$ of all cells between the extreme endpoints form the set
$$
 \calQ=\{h+\tfrac12:h\in\Z,
   \ e_{\min}\le h<e_{\max}\}.
$$
Each interval contains $L$ of these midpoints, so summing their depths counts each selected interval $L$ times:
\begin{equation}
 \sum_{q\in\calQ}\dep_S(q)
 =\sum_{I\in S}|I\cap\calQ|=L|S|.
 \label{eq:constant-depth}
\end{equation}
For a selection of exactly $B$ intervals, this sum is $LB$.

In Figure~\ref{fig:total-depth}, the intervals $I_1=[0,3]$ and $I_2=[1,4]$ both have length $3$. Selecting both gives depths $1,2,2,1$ at the cell midpoints, whose sum is $6=3\cdot2$. Each interval contributes three, including the midpoints in the overlap.

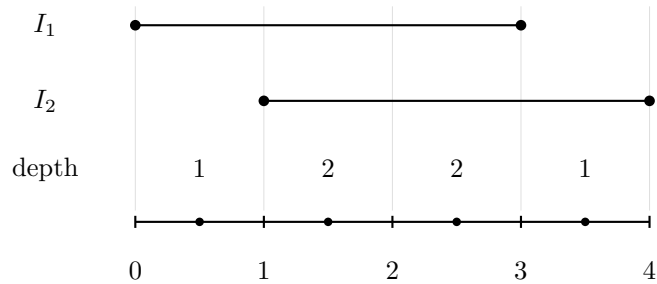
\begin{figure}[htbp]
\centering
\begin{tikzpicture}[x=1.7cm,y=1cm,font=\small,line width=0.8pt]
  \useasboundingbox (-1.3,-1.4) rectangle (4.5,3.5);
  \foreach \h in {0,...,4}
    \draw[gray!25,very thin] (\h,0.15) -- (\h,2.85);
  \draw (0,2.6) -- (3,2.6);
  \fill (0,2.6) circle (2pt);
  \fill (3,2.6) circle (2pt);
  \node at (-0.7,2.6) {$I_1$};
  \draw (1,1.6) -- (4,1.6);
  \fill (1,1.6) circle (2pt);
  \fill (4,1.6) circle (2pt);
  \node at (-0.7,1.6) {$I_2$};
  \draw (0,0) -- (4,0);
  \foreach \h in {0,...,4} {
    \draw (\h,-0.08) -- (\h,0.08);
    \node at (\h,-0.65) {$\h$};
  }
  \foreach \q in {0.5,1.5,2.5,3.5}
    \fill (\q,0) circle (1.6pt);
  \node at (-0.7,0.7) {depth};
  \foreach \q/\z in {0.5/1,1.5/2,2.5/2,3.5/1}
    \node at (\q,0.7) {$\z$};
\end{tikzpicture}
\caption{Total depth.}
\label{fig:total-depth}
\end{figure}

Discard rewards at query points outside the endpoint range, since these points always have depth zero and cannot earn a reward with a positive target depth. Move each remaining query to a half-integer between the same two consecutive interval endpoints, so that the same intervals still contain it. If several queries move to one cell, retain all their rewards there.

Let $\mathcal R$ list all original reward pairs $(w,d)$, including the discarded rewards and counting repeated rewards separately. Let $\mathcal R_q$ list those moved to the cell represented by $q$. For $0\le z\le B$, define the reward in this cell and the total weight of all original rewards by
$$
 f_q(z)=\sum_{(w,d)\in\mathcal R_q}w\ind{z=d},
 \qquad
 \lambda=\sum_{(w,d)\in\mathcal R}w.
$$
A cell with no rewards has $f_q\equiv0$. All target depths are positive, so
\begin{equation}
 f_q(0)=0,\qquad 0\le f_q(z)\le\lambda.
 \label{eq:local-reward-bounds}
\end{equation}

\paragraph{Representing rewards by demand points.}

We add $\lambda z$ in every cell and write $g_q(z)=f_q(z)+\lambda z$. Since $f_q$ can decrease by at most $\lambda$ in one step, $g_q$ is nondecreasing. For $1\le d\le B$, its increase from depth $d-1$ to depth $d$ is
\begin{equation}
 a_{q,d}=g_q(d)-g_q(d-1)
        =\lambda+f_q(d)-f_q(d-1).
 \label{eq:reward-increments}
\end{equation}
By~\eqref{eq:local-reward-bounds},
$$
 a_{q,d}\ge\lambda-f_q(d-1)\ge0.
$$
Thus all $a_{q,d}$ are nonnegative integers. Since $g_q(0)=0$, summing its increases up to depth $z$ gives
\begin{equation}
 g_q(z)=\sum_{d=1}^B a_{q,d}\ind{z\ge d}
 \qquad(0\le z\le B).
 \label{eq:threshold-sum}
\end{equation}
Create $a_{q,d}$ points of demand $d$, each of value one. At depth $z$, the points with $d\le z$ are fully covered, giving $g_q(z)$ points in total.

For $q=h+\tfrac12$, place these points inside $(h,h+1)$. An interval with integer endpoints contains either the whole open cell or none of it, so every point in the cell has depth $\dep_S(q)$. To make their coordinates distinct, set $m_q=\sum_{d=1}^B a_{q,d}$ and put them at
\begin{equation}
 h+\frac{r}{m_q+1},\qquad r=1,\ldots,m_q.
 \label{eq:distinct-points}
\end{equation}

In Table~\ref{tab:reward-conversion}, take $B=3$, $f_q(z)=\ind{z=2}$, and $\lambda=1$. The values $a_{q,1}=1$ and $a_{q,2}=2$ give three points with demands $1,2,2$. Place them at $h+\tfrac14$, $h+\tfrac12$, and $h+\tfrac34$, respectively. At depth one only the first point is fully covered; at depths two and three all three are fully covered.

\begin{table}[htbp]
\centering
\caption{Reward conversion.}
\label{tab:reward-conversion}
\begin{tabular}{@{}lrrrr@{}}
\toprule
$z$ & $0$ & $1$ & $2$ & $3$\\
\midrule
$f_q(z)$ & $0$ & $0$ & $1$ & $0$\\
$g_q(z)$ & $0$ & $1$ & $3$ & $3$\\
$a_{q,z}$ & --- & $1$ & $2$ & $0$\\
\bottomrule
\end{tabular}
\end{table}

For a selection $S$ with $|S|\le B$, define its total reward before this conversion by
$$
 F_{\mathrm{ex}}(S)=\sum_{q\in\calQ}f_q(\dep_S(q)).
$$
The number $F_{\mathrm{cov}}(S)$ of fully covered new points is
$$
 F_{\mathrm{cov}}(S)=\sum_{q\in\calQ}g_q(\dep_S(q)).
$$
For $|S|=B$, equation~\eqref{eq:constant-depth} gives
\begin{align}
 F_{\mathrm{cov}}(S)
 &=\sum_{q\in\calQ}g_q(\dep_S(q))\notag\\
 &=\sum_{q\in\calQ}f_q(\dep_S(q))
        +\lambda\sum_{q\in\calQ}\dep_S(q)\notag\\
 &=F_{\mathrm{ex}}(S)+\lambda LB.
 \label{eq:constant-shift}
\end{align}
For selections of exactly $B$ intervals, the two objectives differ by the fixed value $\lambda LB$.

\begin{lemma}
\label{lem:reward-conversion}
Consider $N$ intervals of common positive integer length $L$, with integer endpoints and a budget $1\le B\le N$. Assume that the distance between the smallest and largest endpoints is bounded by a polynomial in the input size. Each reward has a positive integer weight, is assigned to a query point that is not an interval endpoint, and is obtained at a specified positive integer depth. If the total reward weight $\lambda$ is also bounded by a polynomial in the input size, the construction above gives an unweighted partial interval multicover instance of polynomial size with the same candidate intervals. For every nonnegative integer $Z$,
$$
\begin{aligned}
 &\exists S:\ |S|=B,\quad F_{\mathrm{ex}}(S)\ge Z\\
 &\hspace{1.7em}\Longleftrightarrow
 \exists S:\ |S|\le B,\quad
 F_{\mathrm{cov}}(S)\ge Z+\lambda LB.
\end{aligned}
$$
\end{lemma}
\begin{proof}
If $|S|=B$ and $F_{\mathrm{ex}}(S)\ge Z$, equation~\eqref{eq:constant-shift} gives $F_{\mathrm{cov}}(S)\ge Z+\lambda LB$.

Now suppose that $|S|\le B$ and $F_{\mathrm{cov}}(S)\ge Z+\lambda LB$. Since $B\le N$, extend $S$ to a selection $S'$ of exactly $B$ intervals. Adding intervals cannot decrease the number of fully covered points. Thus $F_{\mathrm{cov}}(S')\ge Z+\lambda LB$, and~\eqref{eq:constant-shift} gives
$$
 F_{\mathrm{ex}}(S')=F_{\mathrm{cov}}(S')-\lambda LB\ge Z.
$$

The assumed bound on the distance between the extreme endpoints gives polynomially many cells. In each cell, the number of new points is
$$
 m_q=\sum_{d=1}^B a_{q,d}=g_q(B)
     =f_q(B)+\lambda B\le\lambda(B+1).
$$
Hence there are polynomially many points in total. Their demands are at most $B$, and the coordinates in~\eqref{eq:distinct-points} have polynomial bit length. Every new point has value one, and every interval has cost one and contributes one to the coverage depth at each point it contains.
\end{proof}

\subsection{Completing the reduction}
\label{sec:hardness}

\begin{proof}[Proof of Theorem~\ref{thm:main}]
We use the usual binary encoding of rational coordinates and integer input parameters. A certificate lists at most $\min\{B,N\}$ distinct interval indices. For each point of $X$, compare its coordinate with the endpoints of the selected intervals, count how many contain it, and test whether its demand is met. We can then check~\eqref{eq:decision} in polynomial time in the input length. Hence the problem is in NP.

For NP-hardness, start with $(G,t)$ from Subsection~\ref{sec:grid}. The construction in Subsection~\ref{sec:windows} uses budget $B=N/2$, and Lemma~\ref{lem:window-value} gives
$$
 \opt_{\mathrm{win}}(B)=MH+C+n+2\alpha(G).
$$
The intervals constructed in Subsection~\ref{sec:intervals} have integer endpoints and length $L=N+1$, with distance $O(N^2)$ between the extreme endpoints. The depth at each query equals its window count, so
$$
 \max_{|S|=B}F_{\mathrm{ex}}(S)
 =\opt_{\mathrm{win}}(B)
 =MH+C+n+2\alpha(G).
$$
Apply Lemma~\ref{lem:reward-conversion} with $\lambda$ equal to the sum of the original reward weights and $Z=MH+C+n+2t$. Set the required number of fully covered points to
\begin{equation}
 p=\lambda LB+MH+C+n+2t.
 \label{eq:target}
\end{equation}

Adding intervals cannot decrease the number of fully covered points, and any selection of at most $B$ intervals can be extended to size $B$. Thus~\eqref{eq:constant-shift} gives
$$
 \max_{|S|\le B}F_{\mathrm{cov}}(S)
 =\max_{|S|=B}F_{\mathrm{cov}}(S)
 =\lambda LB+MH+C+n+2\alpha(G).
$$
The constructed instance has a selection of at most $B$ intervals fully covering at least $p$ points if and only if
$$
 \lambda LB+MH+C+n+2\alpha(G)
 \ge \lambda LB+MH+C+n+2t,
$$
which is equivalent to $\alpha(G)\ge t$. The constant additions are $C$ after optimizing the added path bits, $MH$ for consistent window assignments, and $\lambda LB$ for every selection of $B$ intervals.

We finally bound the output size. The grid construction has $Q=O(n^2)$ and $K,T=O(n)$, so
$$
 N,B,M=O(n^2),\qquad U,H=O(n^2).
$$
The window construction has $O(n^2)$ rewards, each of weight at most $H$, so $\lambda=O(n^4)$. There are $O(N^2)$ cells, each containing at most $\lambda(B+1)$ new points. The total number of demand points is therefore
$$
 O(N^2\lambda B)=O(n^{10}).
$$
Their coordinates and demands have polynomial bit length, and each construction can be performed in polynomial time.

Divide every endpoint and point coordinate by $L$. All intervals now have length one, while each interval contains exactly the same demand points as before. This proves Theorem~\ref{thm:main}.
\end{proof}

\begin{corollary}
\label{cor:optimization}
The problem of selecting the fewest intervals to fully cover at least $p$ points, and the problem of fully covering as many points as possible with at most $B$ intervals, are NP-hard even for unit intervals.
\end{corollary}

\section{Conclusion}

We have proved that unweighted partial interval multicover is NP-complete even for unit intervals, answering the question in~\cite{Li2026}. 

\bibliographystyle{plain}
\bibliography{references}

@article{LI2026,
	title = {Partial interval multicover: Approximation and complexity},
	journal = {Theoretical Computer Science},
	volume = {1075},
	pages = {115958},
	year = {2026},
	issn = {0304-3975},
	author = {Peng Li and Xiangzhi Tu and Zhao Zhang and Yixin Cao},
}

@article{Mohar2001,
	title={Face covers and the genus problem for apex graphs},
	author={Mohar, Bojan},
	journal={Journal of Combinatorial Theory, Series B},
	volume={82},
	number={1},
	pages={102--117},
	year={2001},
	publisher={Elsevier}
}

@article{TamassiaTollis1989,
	title={Planar grid embedding in linear time},
	author={Tamassia, Roberto and Tollis, Ioannis G},
	journal={IEEE Transactions on circuits and systems},
	volume={36},
	number={9},
	pages={1230--1234},
	year={1989},
	publisher={IEEE}
}

@article{Wocjan2003,
	title={Treating the independent set problem by 2d Ising interactions with adiabatic quantum computing},
	author={Wocjan, Pawel and Janzing, Dominik and Beth, Thomas},
	journal={Quantum Information Processing},
	volume={2},
	number={4},
	pages={259--270},
	year={2003},
	publisher={Springer}
}

@article{BogartWest1999,
	title = {A short proof that ‘proper = unit’},
	journal = {Discrete Mathematics},
	volume = {201},
	number = {1},
	pages = {21-23},
	year = {1999},
	issn = {0012-365X},
	author = {Kenneth P. Bogart and Douglas B. West},
}

@article{Slavik1997,
	title={Improved performance of the greedy algorithm for partial cover},
	author={Slav{\'\i}k, Petr},
	journal={Information Processing Letters},
	volume={64},
	number={5},
	pages={251--254},
	year={1997},
	publisher={Elsevier}
}

@article{RajagopalanVazirani1998,
	title={Primal-dual RNC approximation algorithms for set cover and covering integer programs},
	author={Rajagopalan, Sridhar and Vazirani, Vijay V},
	journal={SIAM Journal on Computing},
	volume={28},
	number={2},
	pages={525--540},
	year={1998},
	publisher={SIAM}
}

@article{Chekuri2012,
	title={On the set multicover problem in geometric settings},
	author={Chekuri, Chandra and Clarkson, Kenneth L and Har-Peled, Sariel},
	journal={ACM Transactions on Algorithms (TALG)},
	volume={9},
	number={1},
	pages={1--17},
	year={2012},
	publisher={ACM, New York, NY, USA}
}

@article{Shi2019,
	title={Approximation algorithm for the partial set multi-cover problem},
	author={Shi, Yishuo and Ran, Yingli and Zhang, Zhao and Willson, James and Tong, Guangmo and Du, Ding-Zhu},
	journal={Journal of Global Optimization},
	volume={75},
	number={4},
	pages={1133--1146},
	year={2019},
	publisher={Springer}
}

@article{Ran2021,
	title={Breaking the r max barrier: Enhanced approximation algorithms for partial set multicover problem},
	author={Ran, Yingli and Zhang, Zhao and Tang, Shaojie and Du, Ding-Zhu},
	journal={INFORMS Journal on Computing},
	volume={33},
	number={2},
	pages={774--784},
	year={2021},
	publisher={INFORMS}
}

@article{Golab2009,
	title={Sequential dependencies},
	author={Golab, Lukasz and Karloff, Howard and Korn, Flip and Saha, Avishek and Srivastava, Divesh},
	journal={Proceedings of the VLDB Endowment},
	volume={2},
	number={1},
	pages={574--585},
	year={2009},
	publisher={VLDB Endowment}
}

@article{RanJinZhang2025,
	title={Approximation Algorithm for the Minimum Interval Partial Multi-Cover Problem},
	author={Ran, Yingli and Jin, Jianhong and Zhang, Zhao},
	journal={Networks},
	volume={85},
	number={3},
	pages={288--293},
	year={2025},
	publisher={Wiley Online Library}
}

@article{XuRanZhang2025,
	title={A 1/2-approximation algorithm for maximum interval multi-cover},
	author={Xu, Yuxuan and Ran, Yingli and Zhang, Zhao},
	journal={Theoretical Computer Science},
	volume={1049},
	pages={115369},
	year={2025},
	publisher={Elsevier}
}
\end{document}